\documentclass[a4paper,onecolumn,11pt,accepted=2019-05-27]{quantumarticle}
\pdfoutput=1
\usepackage[utf8]{inputenc}
\usepackage[english]{babel}
\usepackage[T1]{fontenc}
\usepackage{amsmath}
\usepackage{hyperref}

\usepackage{tikz}
\usepackage{lipsum}

\usepackage{amssymb}
\usepackage{amsthm}

\newtheorem{thm}{Theorem}[section]
\newtheorem{lemma}[thm]{Lemma}
\newtheorem{cor}[thm]{Corollary}
\newtheorem{prop}[thm]{Proposition}
\newtheorem{defn}[thm]{Definition}

\def\ket#1{| #1 \rangle}
\def\bra#1{\langle #1 |}

\def\ootimes{\otimes_{\min}}
\def\psii{\psi_{\mbox{\scriptsize\rm initial}}}
\def\psit{\psi_{\mbox{\scriptsize\rm target}}}
\def\half{\frac{1}{2}}
\def \kk#1{| #1 \rangle}                                                            

\newcommand{\cl}[1]{\mathcal{#1}}

\begin{document}

\title{Constant gap between conventional strategies and those based on C*-dynamics for self-embezzlement}

\author{Richard Cleve}
\affiliation{Institute for Quantum Computing and Cheriton School of Computer Science, University of Waterloo, Canada.}
\email{cleve@uwaterloo.ca}
\author{Beno\^{i}t Collins}
\affiliation{Department of Mathematics, Kyoto University, Kyoto 606-8502, Japan.}
\email{collins@math.kyoto-u.ac.jp}
\author{Li Liu}
\affiliation{Institute for Quantum Computing and Cheriton School of Computer Science, University of Waterloo, Canada.}
\email{147liu@uwaterloo.ca}
\author{Vern Paulsen}
\affiliation{Institute for Quantum Computing and Department of Pure Mathematics, University of Waterloo, Canada.}
\email{vpaulsen@uwaterloo.ca}

\maketitle

\begin{abstract}
We consider a bipartite transformation that we call \emph{self-embezzlement} and use it to prove a constant gap between the capabilities of two models of quantum information: the conventional model, where bipartite systems are represented by  tensor products of Hilbert spaces; 
and a natural model of quantum information processing for abstract states on C*-algebras, 
where joint systems are represented by tensor products of C*-algebras.
We call this the \emph{C*-circuit} model and show that it is a special case of the commuting-operator model (in that it can be translated into such a model).
For the conventional model, we show that there exists a constant $\epsilon_0 > 0$ such that self-embezzlement cannot be achieved with precision parameter less than $\epsilon_0$ (i.e., the fidelity cannot be greater than $1 - \epsilon_0$); whereas, in the C*-circuit model---as well as in a commuting-operator model---the precision can be $0$ (i.e., fidelity~$1$).

Self-embezzlement is not a non-local game, hence our results do not impact the celebrated Connes Embedding conjecture.
Instead, the significance of these results is to exhibit a reasonably natural quantum information processing problem for which there is a constant gap between the capabilities of the conventional Hilbert space model and the commuting-operator or C*-circuit model. 
\end{abstract}

\section{Introduction and summary}

In the conventional model of quantum information, separate quantum systems are represented by Hilbert spaces and joint systems are represented by their tensor products.
Localized dynamics and measurements are operations on the Hilbert spaces of the subsystems. 
 
This model---that we refer to here as the \emph{conventional} model---is not fully general.
In a more general \emph{commuting-operator} model, there is one global Hilbert space and the localized dynamics and measurements act on that space with the requirement that certain operators on separate subsystems commute.
In 2017, Slofstra \cite{Slofstra2017} showed that there are \emph{non-local correlations} that can be attained in the commuting-operator model that cannot be obtained in the conventional model.
It remains an open question whether every commuting-operator correlation can be approximated to arbitrary precision by a conventional correlation.
This question is equivalent to that of the celebrated \emph{Connes embedding conjecture}~\cite{Fritz2012,junge2011,Ozawa2013}.

We consider a broader scenario than non-local correlations, and prove that there is a task that can be performed in the commuting-operator model (as well as a model of quantum information processing for abstract states on C*-algebras) that has the property that it cannot even be approximated to arbitrary precision in the conventional model.

Our task is a variant of \emph{embezzlement}, which was introduced by van Dam and Hayden in~\cite{vanDamH03}.
Embezzlement is a mapping where, by local operations, an entangled state $\psi$ is used catalytically to create some other entangled state $\phi$ with high fidelity.
An $\epsilon$-precision embezzlement scheme for state $\phi$ using catalyst $\psi$ with precision parameter $\epsilon$ is a set of local operations that map 
$\psi \otimes (\ket{0} \otimes \ket{0})$ to $\psi \otimes \phi$ within fidelity $1-\epsilon$.
The case where $\epsilon = 0$ corresponds to exact embezzlement.
In the conventional model, any entangled state $\phi$ can be embezzled with precision arbitrarily close to 0 but not exactly (even if the Hilbert spaces of the individual systems are allowed to be infinite dimensional and $\psi$ has infinite entanglement entropy); however, in a commuting-operator model, states can be embezzled exactly~\cite{CleveLP2017}.

Embezzlement cannot be directly tested experimentally the way non-local correlations can, because the parties can utilize concealed entanglement.
Nevertheless, non-local correlations based on the \emph{idea} of embezzlement have been discovered that can approximated to arbitrary precision $\epsilon$ in the conventional model, but where the amount of entanglement required is $\Omega(1/\epsilon)$~\cite{JiLV2018} (see also the earlier related results~\cite{RV2013,LeungTW13}).

Our new result concerns a task that we call \emph{self-embezzlement} and which is remarkable because it can be achieved exactly in the commuting-operator model, whereas \emph{it cannot even be approximated to arbitrary precision} in the conventional model.
In self-embezzlement, a second copy of the catalyst state is embezzled.
That is, by local operations, state $\psi \otimes (\ket{0} \otimes \ket{0})$ is mapped to $\psi \otimes \psi$ within fidelity $1-\epsilon$.
Here, we don't have a specific target state; rather, we allow the catalyst to be 
any pure state that is non-trivially entangled in the sense that it can be used to approximately attain the maximal violation of the CHSH inequality~\cite{ClauserHSH69} (i.e., by a factor of $\sqrt{2} - \epsilon$).
Our main results, stated informally, are: 

\begin{thm}
There exists an $\epsilon_0 > 0$ such that, in the conventional model, approximate self-embezzlement to precision $\epsilon_0$ is impossible.
\end{thm}

\begin{thm}\label{thm:yesgo}
In the commuting-operator model, exact self-embezzlement is possible.
\end{thm}

We do not know whether there are non-local correlations based on the idea of self-embezzlement that exhibit a gap between the two models and make no claim of any consequences regarding the Connes conjecture.
Instead, the significance of these results is to exhibit a reasonably natural quantum information processing problem for which there is a constant gap between the conventional model and the commuting-operator model.
(See~\cite{NavascuesP2012} for another example of a constant-gap separation between the conventional and commuting-operator model---for the task of steering.)

An additional contribution of this work is the proposal and development of a natural model of information processing for abstract states on C*-algebras, that we call the \emph{C*-circuit} model, where the reversible gates are (suitably localized) $*$-automorphisms 
(see Appendix~\ref{sec:c-star-basics} for a brief review of the definitions of \emph{C*-algebra}, \emph{abstract state}, and \emph{$*$-automorphism}).
In fact, our exact self-embezzlement protocol is expressed in the C*-circuit model and can be converted into a commuting-operator model by applying the GNS Theorem \cite{GelfandN1943,Segal1947} to a suitable crossed product of our C*-algebra.
Thus, Theorem~\ref{thm:yesgo} is a corollary of:

\begin{thm}\label{thm:c-star-circuits}
In the C*-circuit model, exact self-embezzlement is possible.
\end{thm}

Theorem~\ref{thm:c-star-circuits} is similar to something pointed out by Keyl \emph{et al.} in~\cite{KeylSW2003}, where they refer to a ``maximally entangled state of infinite entanglement" that is said to be (in comment ``ME~8" of section~5.A) ``unitarily isomorphic to two copies of itself".
It is important for this to be with respect to some notion of local operations, and our definition and construction in terms of C*-circuits is a way of putting this intuition into a rigorous framework.

It is noteworthy that the construction in Theorem~\ref{thm:c-star-circuits} uses the so-called CAR algebra (an acronym of ``canonical anti-commutation relations"), a particular C*-algebra that is \emph{hyperfinite}, which is a property that permits it to be built up in a natural way from finite-dimensional C*-algebras.
It is known that no constant gap can be obtained for non-local correlations with a hyperfinite C*-algebra~\cite{scholz2008}.
That is, for any hyperfinite C*-algebras, $\mathcal{A}$ and $\mathcal{B}$, the non-local correlations attainable with the corresponding C*-circuit model are limit points of non-local correlations attainable in the conventional model.
The proof of this uses the fact that, for hyperfinite C*-algebras, $\mathcal{A} \otimes_{\max} \mathcal{B} = \mathcal{A} \otimes_{\min} \mathcal{B}$.
A consequence of this is that no constant gap between the conventional model and the C*-circuit model for non-local correlations can be obtained by a simple black-box reduction to our self-embezzlement transformation.

Returning to the separation between models in~\cite{Slofstra2017}, this can also be expressed as a difference between the C*-circuit model and the conventional (circuit) model, albeit not by a constant-gap.
In that case the C*-algebras involved in the construction do not appear to be hyperfinite (rather, they are a star-crossed product of the CAR algebra with a group action).

In summary, we obtain a constant-gap separation in capability between the C*-circuit model (with the CAR algebra) and the conventional model for a natural problem.
We believe that the C*-circuit model is a natural model for capturing the capabilities of quantum information processing for infinite-dimensional systems represented as abstract states on C*-algebras, and that its capabilities merit further investigation.

\section{Definitions}

\subsection{The conventional model and the C*-circuit model}

The \emph{quantum circuit} model is the underlying paradigm in which almost all quantum computations and protocols are expressed.

In the \emph{conventional} circuit model, \emph{registers} are represented by Hilbert spaces and \emph{compound registers}%
\footnote{We are using the terminology for registers in~\cite{Watrous2018} (including compound registers).}
 are represented by tensor products of Hilbert spaces.
The \emph{states} of a register are the density operators on its Hilbert space.
The states of (possibly compound) registers can be transformed by a series of reversible \emph{gates}, which are unitary operations on the associated Hilbert spaces; and also they can be \emph{measured} by POVMs (positive-operator valued measures) on the Hilbert spaces.

The \emph{C*-circuit} model is a natural analogue of the conventional model for abstract states on C*-algebras%
\footnote{See Appendix~\ref{sec:c-star-basics} for definitions and some basic properties of C*-algebras and abstract states on them.}.
In the C*-circuit model, \emph{registers} are represented by C*-algebras 
and \emph{compound registers} are represented by tensor products of C*-algebras.
The \emph{states} of a register are the unital positive linear functionals on its C*-algebra.
The states of (possibly compound) registers can be transformed by a series of reversible \emph{gates}, which are $*$-automorphisms of the associated C*-algebras; and also they can be \emph{measured} by POVMs with elements from the C*-algebras.

The following table compares the conventional quantum circuit model and the C*-circuit model.

\begin{figure}[h!]
\small
\renewcommand*{\arraystretch}{1.5}
\begin{tabular}{|l|l|}
\hline
Conventional quantum circuit model & C*-circuit model ($\otimes_{\min}$ version)\\
\hline
\hline
Register $\mathsf{R}$: associated Hilbert space $\mathcal{H}_{\mathsf{R}}$ & 
Register $\mathsf{R}$: associated C*-algebra $\mathcal{A}_{\mathsf{R}}$  
\\
\hline
Compound register $(\mathsf{R}_1,\mathsf{R}_2)$: 
$\mathcal{H}_{\mathsf{R}_1}\otimes \mathcal{H}_{\mathsf{R}_2}$ & 
Compound register $(\mathsf{R}_1,\mathsf{R}_2)$: 
$\mathcal{A}_{\mathsf{R}_1}\ootimes \mathcal{A}_{\mathsf{R}_2}$
\\
\hline
States of $\mathsf{R}$: density operators on $\mathcal{H}_{\mathsf{R}}$ &
States of $\mathsf{R}$: unital positive linear functionals on $\mathcal{A}_{\mathsf{R}}$ \\
\hline
Dynamics of $\mathsf{R}$: unitary operators on $\mathcal{H}_{\mathsf{R}}$ &
Dynamics of $\mathsf{R}$: $*$-automorphisms of $\mathcal{A}_{\mathsf{R}}$
\\
\hline
Measurements of $\mathsf{R}$: POVMs on $\mathcal{H}_{\mathsf{R}}$ &
Measurements of $\mathsf{R}$: POVMs with elements from $\mathcal{A}_{\mathsf{R}}$
\\
\hline
\end{tabular}
\caption{Conventional quantum circuit model vs.\ C*-circuit model}
\end{figure}

The gates acting on a register are the reversible \emph{dynamics} of the register and the $*$-automorphisms are the most general such operations that preserve the algebraic structure and norm of a register's C*-algebra (in analogy with unitary operations in the conventional model, which are the most general operations that preserve the algebraic structure and norm of a register's Hilbert space).
Our framework resembles that of \emph{C*-dynamical systems}, which are C*-algebras combined with sets of $*$-automorphisms acting on them%
\footnote{
More precisely, a \emph{C*-dynamical system} is a triple of the form $(\mathcal{R},G,\alpha)$, where $\mathcal{R}$ is a C*-algebra, $G$ is a locally compact group, and $\alpha$ is a continuous action of $G$ on 
the $*$-automorphisms of $\mathcal{R}$.
}; 
however, in the C*-circuit model there are various localization conditions imposed on the dynamics (i.e., each gate acts on a subset of the registers).

The above definition is a basic model where the gates are reversible and where measurements produce classical outcomes but no residual (or ``collapsed") quantum states.
This model is complete in that channels and measurements that produce residual states can be defined in a Stinespring form (as $*$-automorphisms on a larger system).
Kraus operators (that need not be elements of the C*-algebra) can also be defined, though we do not do that here.

\subsection{Informal definition of self-embezzlement}

Alice and Bob each have two quantum systems, call them $\mathsf{A}_1$, $\mathsf{A}_2$ and $\mathsf{B}_1$, $\mathsf{B}_2$ (respectively).
In computer science terminology, we call these \emph{registers}.
First, we define exact self-embezzlement, followed be $\epsilon$-approximate self-embezzlement.

Definition of \emph{exact self-embezzlement}:
\begin{itemize}
\item
There is some catalyst state $\psi$ that satisfies a \emph{nontriviality} condition that rules out product states and states close to product states.
The condition is that $\psi$ can be used to maximally violate the CHSH inequality by a factor of $\sqrt{2}$.
The catalyst state $\psi$ is allowed to be any pure state that satisfies this property.

\item
The \emph{initial} joint state of Alice and Bob's respective first registers ($\mathsf{A}_1$ and $\mathsf{B}_1$) is the catalyst $\psi$.
The initial joint state of their respective second registers ($\mathsf{A}_2$ and $\mathsf{B}_2$) is some product state, such as $\phi_A\otimes \phi_B$.
So the initial state of $(\mathsf{A}_1,\mathsf{B}_1,\mathsf{A}_2,\mathsf{B}_2)$ is $\psi\otimes(\phi_A\otimes\phi_B)$.

\item
Alice and Bob are each allowed to perform operations that are local to their registers.
For Alice, this is the compound register $(\mathsf{A}_1,\mathsf{A}_2)$.
For Bob, this is the compound register $(\mathsf{B}_1,\mathsf{B}_2)$.
\item
The \emph{final} state of $(\mathsf{A}_1,\mathsf{B}_1,\mathsf{A}_2,\mathsf{B}_2)$ (after they apply their local operations) is $\psi\otimes\psi$.
(That is, $(\mathsf{A}_1,\mathsf{B}_1)$ is in state $\psi$ and $(\mathsf{A}_2,\mathsf{B}_2)$ is in state $\psi$.)

\end{itemize}

\begin{figure}[h!]
\centering
\begin{minipage}{0.85\textwidth}
\centering
\setlength{\unitlength}{0.6mm}
\begin{picture}(9,50)(0,0)
\put(2,45){\makebox(0,0){{\small $\mathsf{A}_1$}}}
\put(2,35){\makebox(0,0){{\small $\mathsf{A}_2$}}}
\put(2,15){\makebox(0,0){{\small $\mathsf{B}_1$}}}
\put(2,5){\makebox(0,0){{\small $\mathsf{B}_2$}}}
\end{picture}
\begin{picture}(85,50)(0,0)
\linethickness{1pt}
\put(20,45){\line(1,0){10}}
\put(20,35){\line(1,0){10}}
\put(20,15){\line(1,0){10}}
\put(20,5){\line(1,0){10}}
\put(30,0){\line(1,0){20}}
\put(30,20){\line(1,0){20}}
\put(30,30){\line(1,0){20}}
\put(30,50){\line(1,0){20}}
\put(30,0){\line(0,1){20}}
\put(30,30){\line(0,1){20}}
\put(50,0){\line(0,1){20}}
\put(50,30){\line(0,1){20}}
\put(50,45){\line(1,0){10}}
\put(50,35){\line(1,0){10}}
\put(50,15){\line(1,0){10}}
\put(50,5){\line(1,0){10}}
\multiput(5,30)(1,1){15}
{\line(0,1){0.4}}
\multiput(5,30)(1,-1){15}
{\line(0,1){0.4}}
\multiput(60,44.8)(1,-1){16}
{\line(0,1){0.4}}
\multiput(60,14.8)(1,1){16}
{\line(0,1){0.4}}
\multiput(60,34.8)(1,-1){16}
{\line(0,1){0.4}}
\multiput(60,4.8)(1,1){16}
{\line(0,1){0.4}}
\put(39,40){\makebox(0,0){{\large $\mathcal O$}{\!$_A$}}}
\put(39,10){\makebox(0,0){{\large $\mathcal O$}{\!$_B$}}}
\put(1.5,30){\makebox(0,0){{\small $\psi$}}}
\put(16.5,35){\makebox(0,0){{\small $\phi$}{\tiny \!$_A$}}}
\put(16.5,5){\makebox(0,0){{\small $\phi$}{\tiny \!$_B$}}}
\put(78.5,30.3){\makebox(0,0){{\small $\psi$}}}
\put(78.5,20.3){\makebox(0,0){{\small $\psi$}}}
\end{picture}
\caption{Circuit diagram for self-embezzlement. In the conventional model, $\mathcal O_A$ and $\mathcal O_B$ are unitary operators acting on registers $(\mathsf{A}_1,\mathsf{A}_2)$ and $(\mathsf{B}_1,\mathsf{B}_2)$, respectively.
In the C*-circuit model, $\mathcal O_A$ and $\mathcal O_B$ are $*$-automorphisms acting on $\mathcal A_1 \otimes_{\min} \mathcal A_2$ and $\mathcal B_1 \otimes_{\min} \mathcal B_2$ (the C*-algebras associated with $(\mathsf{A}_1,\mathsf{A}_2)$ and $(\mathsf{B}_1,\mathsf{B}_2)$), respectively.}
\label{fig:AE}
\end{minipage}
\end{figure}

Next we define \emph{$\epsilon$-approximate self-embezzlement} as the following relaxation of the above.
First of all, the violation of CHSH need only be by a factor of $\sqrt{2} - \epsilon$ (as opposed to the maximum violation of $\sqrt{2}$).
Second, when Alice and Bob apply their local operations to the initial state $\psi\otimes(\phi_A\otimes\phi_B)$, they need only obtain an approximation of the state $\psi\otimes\psi$ within fidelity $1 - \epsilon$.

In the next two subsections, we present precise definitions of approximate self-embezzlement (to match the results in section~\ref{sec:nogo}) and exact self-embezzlement (to match the results in section~\ref{sec:yesgo}).

\subsection{Definition of approximate self-embezzlement in the conventional model}\label{sec:def-AE-mixed-nogo}

\noindent
Define an \emph{$\epsilon$-precision self-embezzling scheme} 
to be a tuple $(\mathcal{H}_A,\mathcal{H}_B,\psi,\phi_A,\phi_B,U_A,U_B)$, where:

\begin{enumerate}
\item
$\mathcal{H}_A$ and $\mathcal{H}_B$ are Hilbert spaces.
(Alice has two registers, which we call $\mathsf{A}_1$ and $\mathsf{A}_2$, associated with $\mathcal{H}_A$, and Bob has two registers, which we call $\mathsf{B}_1$ and $\mathsf{B}_2$, associated with $\mathcal{H}_B$.)
\item
$\psi$ is a normalized vector in $\mathcal{H}_A\otimes\mathcal{H}_B$ such that state $\psi$ can be used to violate the CHSH inequality by factor $\sqrt{2} - \epsilon$.
We call $\psi$ the \emph{catalyst}.
\item
$\phi_A$ and $\phi_B$ are normalized vectors in $\mathcal{H}_A$ and $\mathcal{H}_B$ respectively (with no restriction).
\item
$U_A$ is a unitary operation on $\mathcal{H}_A\otimes\mathcal{H}_A$ and $U_B$ is a unitary operator on $\mathcal{H}_B\otimes\mathcal{H}_B$.
Applying $U_A\otimes U_B$ to system $((\mathsf{A}_1,\mathsf{A}_2),(\mathsf{B}_1,\mathsf{B}_2))$ has the following property:
in the context of system $(\mathsf{A}_1,\mathsf{B}_1,\mathsf{A}_2,\mathsf{B}_2)$, it maps state $\psi\otimes(\phi_A\otimes\phi_B)$ to state $\psi\otimes\psi$ within fidelity $1-\epsilon$.
\end{enumerate}

\subsection{Definition of exact self-embezzlement in the C*-circuit model}\label{sec:def-AE-pure-yesgo}

Define an \emph{self-embezzling scheme} 
to be a tuple of the form $(\mathcal{A},\mathcal{B},\psi,\phi_A,\phi_B,\alpha_A,\alpha_B)$, where:
\begin{enumerate}
\item
$\mathcal{A}$ and $\mathcal{B}$ are C*-algebras.
(Alice has two registers, which we call $\mathsf{A}_1$ and $\mathsf{A}_2$, associated with $\mathcal{A}$, and Bob has two registers, which we call $\mathsf{B}_1$ and $\mathsf{B}_2$, associated with $\mathcal{B}$.)
\item
$\psi : \mathcal{A}\otimes_{\min}\mathcal{B} \rightarrow \mathbb{C}$ can be any pure state which has the property that $\psi$ can be used to maximally violate the CHSH inequality%
\footnote{CHSH in the framework where Alice is allowed to perform POVM measurements with elements in $\mathcal{A}$, and similarly for Bob with $\mathcal{B}$.}
(by factor $\sqrt{2}$).
We call $\psi$ the \emph{catalyst}.
\item
$\phi_A : \mathcal{A} \rightarrow \mathbb{C}$ and $\phi_B : \mathcal{B} \rightarrow \mathbb{C}$ are states (with no restriction).
\item
$\alpha_A$ is a $*$-automorphism on $\mathcal{A}\otimes_{\min}\mathcal{A}$ and $\alpha_B$ is a $*$-automorphism on $\mathcal{B}\otimes_{\min}\mathcal{B}$.
\item\label{enumerate:def}
Applying $\alpha_A\otimes\alpha_B$ to system $(\mathsf{A}_1,\mathsf{A}_2,\mathsf{B}_1,\mathsf{B}_2)$ has the following property:
in the context of system $(\mathsf{A}_1,\mathsf{B}_1,\mathsf{A}_2,\mathsf{B}_2)$, it maps state $\psi\otimes(\phi_A\otimes\phi_B)$ to state $\psi\otimes\psi$.
\end{enumerate}

\section{There exists $\epsilon_0 > 0$ such that, in the conventional model, self-embezzlement with precision $\le \epsilon_0$ is impossible}\label{sec:nogo}

\subsection{The case of unitary operations}

Without loss of generality, a pure catalyst state is of the form 
\begin{align}\label{eq:tensor}
\psi
= \sum_{k=1}^{\infty} \lambda_k \ket{k}\otimes\ket{k},
\end{align}
where $\lambda_1 \ge \lambda_2 \ge \cdots$ and $\sum_k |\lambda_k|^2 = 1$, 
and the initial states of $\phi_A$ and $\phi_B$ are $\ket{1}$.

The initial state of $(\mathsf{A}_1,\mathsf{B}_1,\mathsf{A}_2,\mathsf{B}_2)$ is 
\begin{align}
\psii
= \psi\otimes(\phi_A\otimes\phi_B)
= \left(\sum_{k=1}^{\infty} \lambda_k \ket{k}\otimes\ket{k}\right)\otimes \Bigl(\ket{1}\otimes\ket{1}\Bigr).
\end{align}
In a purported self-embezzlement scheme, Alice applies a local unitary $U_A$ on register $(\mathsf{A}_1,\mathsf{A}_2)$ and Bob $U_B$ on register $(\mathsf{B}_1,\mathsf{B}_2)$.
We shall bound the fidelity between $(U_A \otimes U_B)\psii$ and

\begin{align}\label{eq:psixpsiV0}
\psit
= \psi \otimes \psi
= \Biggl(\sum_{j=1}^{\infty} \lambda_j \ket{j}\otimes\ket{j}\Biggr)\otimes\left(\sum_{k=1}^{\infty} \lambda_k \ket{k}\otimes\ket{k}\right).
\end{align}

\begin{thm}\label{thm:main-nogo}
There exists a constant $\epsilon_0 > 0$ such that, for any $\psi$ that is $(\sqrt{2}-\epsilon_0)$-CHSH violating, for any local unitary operations $U_A$ and $U_B$, the trace distance between $(U_A \otimes U_B)\psii$ and $\psit$ is at least~$\frac{2}{9}$.
\end{thm}

\begin{proof}
If $\psi$ can be used to violate the CHSH inequality by a factor of $\sqrt{2}-\epsilon$ then, by the rigidity results in~\cite{ReichardtUV2013}, there exist local unitary operations that map  $\psi$ within distance $O(\sqrt{\epsilon})$ from a state of the form $(\frac{1}{\sqrt 2}\ket{00}+\frac{1}{\sqrt 2}\ket{11})\otimes \psi'$.
This implies that the largest Schmidt coefficient of $\psi$ satisfies $\lambda_1 \le \frac{1}{\sqrt 2}+O(\sqrt{\epsilon})$.
Set $\epsilon_0 > 0$ to be sufficiently small%
\footnote{Using the bounds in~\cite{Kaniewski2016}, it can be calculated that it suffices to set $\epsilon_0 \le \frac{1}{50}$.}
so that $\lambda_1 \le \sqrt{2/3}$.

Using the fact that $\lambda_1 \le \sqrt{2/3}$, we shall show that, for any local unitaries $U_A$ and $U_B$, the trace distance between $(U_A \otimes U_B)\psii$ and $\psit$ is at least $\frac{2}{9}$.

Expressed as states on $((\mathsf{A}_1,\mathsf{A}_2),(\mathsf{B}_1,\mathsf{B}_2))$, the initial state and target states are 
\begin{align}
\psii
&= \sum_{j=1}^{\infty}\sum_{k=1}^{\infty} \lambda_j \delta_{k,1} 
\Bigl(\ket{j}\otimes\ket{k}\Bigr)\otimes\Bigl(\ket{j}\otimes\ket{k}\Bigr)\label{eq:initial}\\
\psit
&= \sum_{j=1}^{\infty}
\sum_{k=1}^{\infty} \lambda_j\lambda_k 
\Bigl(\ket{j}\otimes\ket{k}\Bigr)\otimes\Bigl(\ket{j}\otimes\ket{k}\Bigr),\label{eq:target}
\end{align}
where $\delta_{i,j}$ is the Kronecker-delta function.

By Lemma~1 in~\cite{VidalJN2000}, the fidelity is maximized when 
the Schmidt-basis states are the same and the Schmidt coefficients are lined up in terms of magnitude (i.e., largest with largest, second largest with the second largest, etc.).
Thus, we need only consider unitary operations $U_A$ and $U_B$ that each apply a permutation of the basis states $\{\ket{j}\otimes\ket{k}\}_{j,k}$ ($U_A$ on register $(\mathsf{A}_1,\mathsf{A}_2)$ and $U_B$ the same permutation on register $(\mathsf{B}_1,\mathsf{B}_2)$).

To analyze this, we first consider a related statement about probability distributions.
Let $p = (p_1, p_2, \dots)$ be a probability distribution on $\mathbb N$ (possibly of finite support) and assume that $p_1 \ge p_2 \ge \dots$.
We consider how close a rearrangement of the probabilities in $p$ can be to $p \otimes p$.
A rearrangement can be defined as a bijection $\pi : \mathbb{N}\times\mathbb{N} \rightarrow \mathbb{N}$ where $\pi p$ is the probability distribution on $\mathbb{N}\times\mathbb{N}$ defined as $(\pi p)(j,k) = p_{\pi(j,k)}$.

Recall that the \emph{variation distance} between distributions $p$ and $q$ is defined as $\half\|p-q\|_1 = \half\sum_k |p_k-q_k|$.

\begin{lemma}\label{thm:half}
If $p_1 \le \frac{2}{3}$ and $\pi : \mathbb{N}\times\mathbb{N} \rightarrow \mathbb{N}$ is any bijection then the variation distance between $\pi p$ and $p\otimes p$ is at least $2/9$.
\end{lemma}

\begin{proof}[Proof (of Lemma~\ref{thm:half}).]
Define $m = \max\{m \in \mathbb{N} : p_1 + \dots + p_m \le \frac {2}{3}\}$ and 
$S = \{1,\dots,m\}$.
Then
\begin{align}
\frac{1}{3} < p(S) \le \frac{2}{3}.
\end{align}
The first inequality follows because, if $p_1 + \cdots + p_m \le \frac{1}{3}$ then $p_{m+1} \le \frac{1}{3}$, so $p_1 + \cdots + p_{m+1} \le \frac{2}{3}$.

Define $\mu = p(S)$.
Consider the set $T$, defined as the $m$ largest components of $p\otimes p$.
We next show that $p(T) \le \mu^2$.
To see why this is so, note that 
\begin{align}
(p\otimes p)\bigl(\{1,\dots,m\}\times\{1,\dots,m\}\bigr) = 
(p_1 + p_2 + \cdots + p_m)(p_1 + p_2 + \cdots + p_m) =
\mu^2
\end{align}
and also that it is straightforward to show that 
\begin{align}
T \subseteq \{(j,k) \in \mathbb{N}\times\mathbb{N} : j+k \le m+1\}
\subseteq \{1,\dots,m\}\times\{1,\dots,m\}.
\end{align}
This is because, if $(j',k') \not\in \{(j,k) \in \mathbb{N}\times\mathbb{N} : j+k \le m+1\}$ then there are at least $j'k'-1$ elements of $T$ that are larger than $p_{j'}p_{k'}$.

Therefore, the variation distance between the $m$ largest components of $p \otimes p$ and $\pi p$ is at least $\half(\mu - \mu^2) = \half\mu(1-\mu) \ge \half(\frac{2}{3})(\frac{1}{3}) = \frac{1}{9}$.
It follows that the variation distance between (all components of) 
$p\otimes p$ and $\pi p$ is at least $\half(\mu - \mu^2) + \half((1-\mu^2)-(1-\mu))=\frac{2}{9}$.
This completes the proof of Lemma~\ref{thm:half}.
\end{proof}

Returning to the proof of Theorem~\ref{thm:main-nogo}, Lemma~\ref{thm:half} implies that if $U_A$ and $U_B$ are permutations of the basis states of $\psii$ then one particular way of distinguishing between 
$(U_A\otimes U_B)\psii$ and $\psit$ (based on first measuring the state in the computational basis) distinguishes with probability at least $\frac{2}{9}$.
This implies that the trace distance between $(U_A\otimes U_B)\psii$ and $\psit$ is at least $\frac{2}{9}$.
This completes the proof of Theorem~\ref{thm:main-nogo}.
\end{proof}

\begin{cor}
There exists a constant $\epsilon_0 > 0$ such that, for any $\psi$ that is $(\sqrt{2}-\epsilon_0)$-violating, for any local unitary operations $U_A$ and $U_B$, the fidelity between $(U_A \otimes U_B)\psii$ and $\psit$ is at most $\sqrt{1-(2/9)^2} < 0.974996 < 39/40$.
\end{cor}

\def \HH {\mathcal H}
\def \I {\mathcal I}
\def \swap {\text{swap}}
\subsection{The case of channels} 
It turns out that even if Alice and Bob are allowed to use channels (completely positive trace preserving maps) instead of unitaries, they are still not able to perform approximate self-embezzlement.
More formally, using the same notation as the previous subsection, and let $\mathcal N_A$ be a channel on $(\mathsf A_1, \mathsf A_2)$, $\mathcal N_B$ be a channel on $(\mathsf B_1,\mathsf B_2)$, we have 
\begin{lemma}
  \label{lemma:pure_channel_nogo}
  If $\psi$ is a $(\sqrt 2 - \epsilon_0)$-CHSH violating state for some constant $\epsilon_0>0$, then there exist some threshold $\epsilon > 0$, where no local channels $\mathcal N_A$ and $\mathcal N_B$ can achieve 
  \begin{align}\label{eq:pure_channel_nogo}
    \langle{\psi_{\text{\rm target}}},\,\mathcal N_A\otimes\mathcal N_B(\psi_{\text{\rm initial}} \psi_{\text{\rm initial}}^*){\psi_{\text{\rm target}}\rangle} &>1-\epsilon.
  \end{align}
\end{lemma}
\begin{proof}
  
The proof is by contradiction. Let $U_A$ and $U_B$ be the Stinespring form of $\mathcal N_A$ and $\mathcal N_B$, and call the registers holding the extra qubits from Stinespring dilation $\mathsf P_1$ and $\mathsf P_2$. Let ${\psi_{\text{initial}}'}  = {\psi_{\text{initial}}} \otimes\ket 0\otimes\ket 0$ be $\psi_{\text{initial}}$ extended to $\mathsf P_1$ and $\mathsf P_2$ with $\ket 0$, such that tracing out $\mathsf P_1$ and $\mathsf P_2$ on $U_A\otimes U_B {\psi_{\text{initial}}'}$ gives us $ \mathcal N_A\otimes\mathcal N_B({\psi_{\text{initial}}}{\psi_{\text{initial}}^*})$. Let ${\psi_{\text{final}}} = U_A\otimes U_B  {\psi_{\text{initial}}'}$. If Eq.~\eqref{eq:pure_channel_nogo} holds, then 
  \begin{align}\label{eq:pure_channel_trace}
  \langle{\psi_{\text{target}}},\operatorname{Tr}_{\mathsf P_1, \mathsf P_2}({\psi_{\text{final}}}{\psi_{\text{final}}^*}){\psi_{\text{target}}\rangle}& > 1-\epsilon
\end{align}
  We show that since the partial trace of ${\psi_{\text{final}}}$ is close to ${\psi_{\text{target}}}$, there exists a state $\phi$ in register $(\mathsf P_1, \mathsf P_2)$ such that $|\langle{\psi_{\text{final}}}, (\phi\otimes{\psi_{\text{target}}})\rangle|^2 > 1-2\epsilon$ with the following proposition.

  \begin{prop}
    Let $\psi\in\HH$ and $\phi\in \HH'\otimes\HH$, where 
    \begin{align}
      \langle\psi,\operatorname{Tr}_{\HH'} (\phi \phi^*)\psi\rangle > 1-\epsilon.
  \end{align}
    Then there exists ${\psi_0}\in\HH'$ such that
    \begin{align}
      |\langle\phi, ({\psi_0}\otimes{\psi})\rangle|^2 > 1-2\epsilon
    \end{align}
  \end{prop}
  \begin{proof}
  Let ${\phi} = \sum_i \sqrt{p_i}i\otimes \ket{v_i}$ be a Schmidt decomposition of $\phi$ across $\HH'$ and $\HH$, where $\ket i\in\HH'$. Then we have 
    \begin{align}
      \sum_i p_i \langle\psi ,v_i\rangle \langle v_i,{\psi}\rangle > 1-\epsilon.
    \end{align}
    Without loss of generality, assume $\ket{v_0}$ is the state closest to ${\psi}$, so that $|\langle{\psi}, v_0\rangle|^2 > 1-\epsilon$. Since $\{\ket{v_i}\}$'s are orthonormal to each other, $\sum_{i\neq 0}|{\langle\psi,} v_{i}\rangle|^2 < \epsilon$. We use this to get a bound on $p_0$:
  \begin{align}
    1-\epsilon &< p_0 |\langle{\phi}, v_0\rangle|^2 +\sum_{i>0} p_i |\langle{\psi}, v_i\rangle |^2\leq  p_0 + (1 - p_0) \epsilon = p_0(1 - \epsilon) + \epsilon,
  \end{align}
    therefore $p_0 > \frac{1-2\epsilon}{1-\epsilon}$. Now let ${\psi_0}= \ket 0$, then
$|\langle{\phi}, ({\psi_0}\otimes{\psi})\rangle|^2 = p_0 \langle{v_0},\psi\rangle  > 1-2\epsilon$.
\end{proof}
 
Returning to the proof of Lemma~\ref{lemma:pure_channel_nogo}, since $U_A$ and $U_B$ are local unitary operations for Alice and Bob, the Schmidt coefficients of ${\psi_{\text{final}}}$ are the same as the Schmidt coefficeint of ${\psi_{\text{initial}}}$, which are $\{\lambda_1,\lambda_2,\dots\}$.
Let $\{\gamma_1,\gamma_2,\dots\}$ be the Schmidt coefficients of $\phi$ across register $\mathsf P_1$ and $\mathsf P_2$.
Then the Schmidt coefficient of $\phi\otimes{\psi_{\text{target}}}$ across Alice and Bob's registers are $\{\lambda_i \lambda_j \gamma_k\}_{i,j,k}$.

Since, for each $k$, $0\le \gamma_k \le 1$, the largest $m$ entries of $\{\lambda_i\lambda_j \gamma_k\}_{i,j,k}$ must be less or equal to the largest $m$ entries of $\{\lambda_i\lambda_j\}_{i,j}$ for any $m > 0$. Following the same the proof of Lemma~\ref{thm:half}, if $\psi$ is $(\sqrt 2-\epsilon_0)$-violating for some constant $\epsilon_0>0$, it is not possible to have the fidelity between ${\psi_{\text{final}}}$ and $\phi\otimes{\psi_{\text{target}}}$ be $1-2\epsilon$ for arbitrary $\epsilon > 0$.
\end{proof}

\section{In C*-circuit model, exact self-embezzlement is achievable}\label{sec:yesgo}

\subsection{The CAR algebra}

The C*-algebra used is the so-called CAR algebra (where CAR is an abbreviation of ``canonical anti-commutation relations").
(See~\cite{Da,Pedersen1979} for more background information.)

To define the CAR algebra, we can start with the infinite tensor products of Pauli operators of finite weight, where the Pauli operators are 
$I = (\begin{smallmatrix}
1 & 0 \\ 0 & 1
\end{smallmatrix})$, 
$X = (\begin{smallmatrix}
0 & 1 \\ 1 & 0
\end{smallmatrix})$, 
$Z = (\begin{smallmatrix}
1 & \phantom{-}0 \\ 0 & -1
\end{smallmatrix})$, 
$XZ = (\begin{smallmatrix}
0 & -1 \\ 1 & \phantom{-}0
\end{smallmatrix})$
and the weight of such an infinite tensor product is the number of instances of $X$, $Z$, or $XZ$.
For example, $I \otimes X \otimes (XZ) \otimes I \otimes Z \otimes I \otimes I \otimes \cdots$ has weight 3.
We can denote each such operator as $X^aZ^b$, where $a, b \in \{0,1\}^*$, where it is understood that each string is padded on the right with an infinite sequence of $0$s.
Thus, $X^a Z^b = (X^{a_1}\otimes X^{a_2} \otimes \cdots)(Z^{b_1}\otimes Z^{b_2} \otimes \cdots)$.
The above example is $X^a Z^b$, where $a = 011 \equiv 011000\dots$ and $b = 00101 \equiv 00101000\dots$.
Define the set of generators $G = \{X^a Z^b : a, b \in \{0,1\}^*\}$ and $\mathbb{C}G$ to be the set of all (finite) linear combinations%
\footnote{This is well-defined because there are finitely many terms, each of which has all but finitely many factors of $I$.}
of elements of $G$.
$\mathbb{C}G$ is closed under multiplication and is a $*$-algebra.
(Note that $\{\pm X^a Z^b : a, b \in \{0,1\}^*\} \subset \mathbb{C}G$ is a multiplicative group that we can think of as an infinite version of the Pauli group; however, $G$ itself is not closed under multiplication.)

For each element $A \in \mathbb{C}G$, there is an $m \in \mathbb{N}$ and $M \in \mathbb{C}^{2^m \times 2^m}$ such that $A = M \otimes I \otimes I \otimes \cdots$.
Define a norm on $\mathbb{C}G$ as $\|A\| = \|M\|$ (i.e., the spectral norm of $M$ as an operator on $\mathbb{C}^{2^m}$).
The CAR algebra is the completion of $\mathbb{C}G$ with respect to this norm.

\subsubsection{Notation for the CAR algebra and some of its basic properties}

Henceforth we denote the CAR algebra by $\mathcal{R}$.

Also, since $\mathcal{R} \otimes_{\min} \mathcal{R} = \mathcal{R} \otimes_{\max} \mathcal{R}$ (a consequence of $\mathcal{R}$ being hyperfinite~\cite{Da}), we can unambiguously refer to the C*-algebraic tensor product as $\mathcal{R}\otimes\mathcal{R}$.

Note that, in the aforementioned description of the elements of the generating set $G$ as $X^a Z^b$, we have used $\mathbb{N}$ as the index set for the bits of $a = a_1 a_2 \dots$ and $b = b_1 b_2 \dots$; however, any countably infinite set may be used.
It is sometimes convenient to use $\mathbb{Z}$ as the index set, which corresponds to thinking of the infinite tensor products of Paulis as \emph{two-way infinite} strings.
(An alternative way of thinking about the equivalence between using $\mathbb{N}$ and $\mathbb{Z}$ as the index sets in the specification of the generators of $\mathcal{R}$ is as a $*$-isomorphism between $\mathcal{R}$ and $\mathcal{R} \otimes \mathcal{R}$, where the index set can be $\{0,1,2,\dots\}$ for the first copy of $\mathcal{R}$ and $\{-1,-2,\dots\}$ for the second copy.)

An example of a \emph{$*$-automorphism} $\alpha : \mathcal{R} \rightarrow \mathcal{R}$ is conjugation by some unitary $u \in \mathcal{R}$ (where unitary means $u u^* = u^* u = I$).
That is, $\alpha_u (a) = u^*a u$.
These are called \emph{inner} automorphisms.
Automorphisms that are not inner are called \emph{outer} automorphisms.
An example of an outer automorphism is the bilateral-shift operation that maps $X^a Z^b$ (where $a, b :\mathbb{Z} \rightarrow \{0,1\}$) to $X^{a'} Z^{b'}$, where $a'_j = a_{j+1}$ and $b'_j = b_{j+1}$.
Note that \emph{any} permutation of the index set corresponds to a $*$-automorphism.

\subsection{The self-embezzlement scheme}

We first define a state that can be intuitively thought of as a countably infinite tensor product of states of the form $\Psi = \frac{1}{\sqrt 2}\ket{0}\otimes\ket{0} + \frac{1}{\sqrt 2}\ket{1}\otimes\ket{1}$.
It is impossible to express such a state as a vector in the tensor product of two Hilbert spaces (even if the Hilbert spaces are allowed to have uncountably infinite dimension; a proof of this is in~\cite{CleveLP2017}, whose Appendix~A shows that states in the tensor product of two Hilbert spaces have a Schmidt decomposition, with a countably number of Schmidt coefficients).
 However, as was essentially pointed out in~\cite{KeylSW2003}, such a state \emph{can} be defined as an abstract state 
 $s_{\Psi} : \mathcal{R}\otimes\mathcal{R} \rightarrow \mathbb{C}$ such that
\begin{align}
s_{\Psi}((X^aZ^b)\otimes(X^{a'}Z^{b'})) 
= \prod_{j=-1}^{-\infty}\bra{\Psi}(X^{a_j}Z^{b_j})\otimes(X^{a_j'}Z^{b_j'})\ket{\Psi}
= \prod_{j=-1}^{-\infty}\delta_{a_j,a'_j}\delta_{b_j,b'_j},
\end{align}
where $\delta$ is the Kronecker delta function (and, solely for convenience later on, we are using $-\mathbb{N} = \{-1,-2,\dots\}$ as the index set).
In~\cite{KeylSW2003}, such a state is described as an example of the notion of an ``infinitely entangled state" and several of the properties of this state are explained.
By Proposition~\ref{prop:pure}, the abstract state $s_{\Psi}$ is a pure state.

We also define an abstract state 
$s_{00} : \mathcal{R}\otimes\mathcal{R} \rightarrow \mathbb{C}$ 
that corresponds to an infinite tensor product of $\ket{00} = \ket{0}\otimes\ket{0}$ states as
\begin{align}
s_{00}((X^aZ^b)\otimes(X^{a'}Z^{b'})) 
= \prod_{j=0}^{\infty}\bra{00}(X^{a_j}Z^{b_j})\otimes(X^{a_j'}Z^{b_j'})\ket{00}
= \prod_{j=0}^{\infty}(1-a_j)(1-a'_j).
\end{align}

We set the catalyst state, the initial state of $(\mathsf{A}_1,\mathsf{B}_1)$, to be the combination of $s_{00}$ and $s_{\Psi}$, expressed as $\psi : \mathcal{R}\otimes\mathcal{R} \rightarrow \mathbb{C}$, where 
\begin{align}
\psi((X^aZ^b)\otimes(X^{a'}Z^{b'})) 
= \prod_{j=0}^{\infty}(1-a_j)(1-a'_j)
\prod_{j=-1}^{-\infty}\delta_{a_j,a'_j}\delta_{b_j,b'_j}.
\end{align}
A schematic picture of $\psi$ is illustrated in Figure~\ref{fig:s_psi}.

We set the initial state of $\mathsf{A_2}$ and of $\mathsf{B_2}$ to each be 
$\phi : \mathcal{R} \rightarrow \mathbb{C}$, defined as 
\begin{align}
\phi(X^aZ^b) 
= \prod_{j=-\infty}^{\infty}\bra{0}(X^{a_j}Z^{b_j})\ket{0}
= \prod_{j=-\infty}^{\infty}(1-a_j).
\end{align}
Thus, the initial state of $(\mathsf{A}_2,\mathsf{B}_2)$ is the product state $(\phi\otimes\phi) : \mathcal{R}\otimes\mathcal{R} \rightarrow \mathbb{C}$, where
\begin{align}
(\phi\otimes\phi)((X^aZ^b)\otimes(X^{a'}Z^{b'})) 
= \prod_{j=-\infty}^{\infty}\bra{00}(X^{a_j}Z^{b_j}\otimes X^{a'_j}Z^{b'_j})\ket{00}
= \prod_{j=-\infty}^{\infty}(1-a_j)(1-a'_j).
\end{align}
A schematic picture of $\phi\otimes\phi$ is illustrated in Figure~\ref{fig:s_00}.

\begin{figure}[h!]
\centering
\begin{minipage}{0.4\textwidth}
\centering
\includegraphics[width=0.7\textwidth]{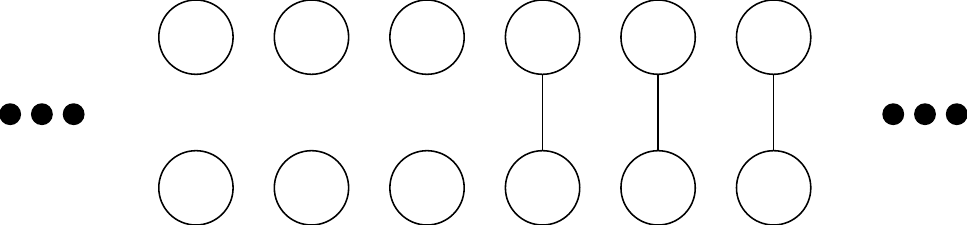}
\vspace{-0.5em}
\caption{schematic of $\psi$, the initial state of $(\mathsf{A}_1,\mathsf{B}_1)$, with infinitely many $\kk{00}$ states on the left and infinitely many $\Psi$ states on the right.}
\label{fig:s_psi}
\end{minipage}
\hspace{0.05\textwidth}
\begin{minipage}{0.4\textwidth}
\centering
\includegraphics[width=0.7\textwidth]{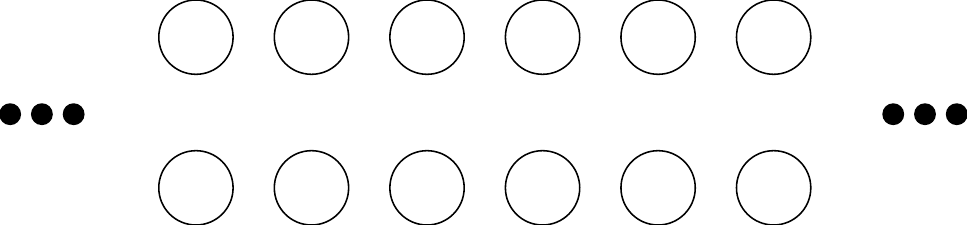}
\vspace{-0.5em}
\caption{schematic of $\phi \otimes \phi$, the initial state of $(\mathsf{A}_2,\mathsf{B}_2)$, with infinitely many $\kk{00}$ states on the left and the right.
This is a product state over $(\mathsf{A}_2,\mathsf{B}_2)$.}
\label{fig:s_00}
\end{minipage}
\end{figure}

The self-embezzlement scheme is based on a $*$-automorphism $\alpha : \mathcal{R} \otimes \mathcal{R} \rightarrow \mathcal{R} \otimes \mathcal{R}$, where 
the same $\alpha$ is applied to $(\mathsf{A}_1,\mathsf{A}_2)$ as well as to $(\mathsf{B}_1,\mathsf{B}_2)$.

Intuitively, $\alpha$ moves infinitely many of the entangled pairs from $(\mathsf{A}_1,\mathsf{B}_1)$ to $(\mathsf{A}_2,\mathsf{B}_2)$ while preserving the entangled pairs in $(\mathsf{A}_1,\mathsf{B}_1)$.
The $\ket{00}$ states are present so that this can be accomplished by a permutation of the qubits.

Formally, $\alpha$ permutes the generators of $\mathcal{R}\otimes\mathcal{R}$, 
which are each of the form $X^a Z^b \otimes X^c Z^d$.
We define $\alpha$ as a reordering of the bits of $(a,c)$ and of the bits of $(b,d)$.
Such a reordering is a permutation on the index set.
The index set of $(a,c)$ is two copies of $\mathbb{Z}$, which can be denoted as $\mathbb{Z}\times\{1,2\}$.
Similarly, for the index set of $(b,d)$.
By a cardinality argument, there exists a bijection $p : \mathbb{Z}\times\{1,2\} \rightarrow \mathbb{Z}\times\{1,2\}$ such that 
\begin{align}
p(\{-1,-2,\dots\}\times\{1\}) = p(\{-1,-2,\dots\}\times\{1,2\}).
\end{align}
We define $\alpha(X^a Z^b \otimes X^c Z^d) = X^{a'} Z^{b'} \otimes X^{c'} Z^{d'}$ where $(a',c')$ is the permutation of the bits of $(a,c)$
corresponding to $p$, and $(b',d')$ is also the permutation of the bits of $(b,d)$ corresponding to $p$.
(This can be expressed as $(a',c')_{\ell} = (a,c)_{p(\ell)}$ and $(b',d')_{\ell} = (b,d)_{p(\ell)}$, for all $\ell \in \mathbb{Z}\times\{1,2\}$.)

Since any permutation on the indices is a $*$-automorphism, $\alpha$ is a $*$-automorphism on $\mathcal{R} \otimes \mathcal{R}$.
Also, $\alpha$ has been designed so that applying this $\alpha$ to $(\mathsf{A}_1,\mathsf{A}_2)$ and to $(\mathsf{B}_1,\mathsf{B}_2)$ transforms $\psi \otimes (\phi \otimes \phi)$ to $\psi \otimes \psi$.

\section{Exact self-embezzlement in the commuting-operator model}

The purpose of this section is to make the statement of Theorem~1.2 more precise.

We begin by sketching a definition of a multi-register commuting-operator framework.
There is one Hilbert space $\mathcal{H}$ and, for registers $\mathsf{R}_1, \dots, \mathsf{R}_m$, there are corresponding algebras of observables, which are C*-algebras $\mathcal{A}_1,\dots,\mathcal{A}_m \subseteq B(\mathcal{H})$, with the requirement that, for each $j \neq k$, $\mathcal{A}_j$ and $\mathcal{A}_k$ commute.
Intuitively, if Alice has register $\mathsf{R}_j$ in her lab then she can perform any POVM measurement with elements in $\mathcal{A}_j$.
Moreover, for a compound register of the form $(\mathsf{R}_{k_1},\dots,\mathsf{R}_{k_{\ell}})$, the associated algebra of observables is defined as $\mathcal{A}_{(k_1,\dots,k_{\ell})} = \overline{\mathcal{A}_{k_1} \cup \cdots \cup \mathcal{A}_{k_{\ell}}}$ (where the bar denotes the C*-algebraic closure).
If the global state of the system is $\psi \in \mathcal{H}$ then the state of  register $(\mathsf{R}_{k_1},\dots,\mathsf{R}_{k_{\ell}})$ is $s : \mathcal{A}_{(k_1,\dots,k_{\ell})} \rightarrow \mathbb{C}$ defined as $s(A) = \langle \psi, A \psi \rangle$ (for all $A \in \mathcal{A}_{(k_1,\dots,k_{\ell})}$).

We can incorporate the reversible dynamics of a register in this model.
Define a unitary $U \in B(\mathcal{H})$ to \emph{act on register $(\mathsf{R}_{k_1},\dots,\mathsf{R}_{k_{\ell}})$} if:
\begin{itemize}
\item
$U$ is \emph{localized} to register $(\mathsf{R}_{k_1},\dots,\mathsf{R}_{k_{\ell}})$.
Technically this is stated as, for all $j \not\in \{k_1,\dots,k_{\ell}\}$ and $A \in \mathcal{A}_j$, $U^*AU = A$.
(In other words, for all for all $j \not\in \{k_1,\dots,k_{\ell}\}$, $U$ centralizes $\mathcal{A}_j$.)

\item
$U$ \emph{preserves the C*-algebra of $(\mathsf{R}_{k_1},\dots,\mathsf{R}_{k_{\ell}})$}.
Since applying $U$ to a state is equivalent to conjugating the POVM elements of a subsequent measurement by $U$, this requirement is technically expressed as 
$U^* \mathcal{A}_{(k_1,\dots,k_{\ell})} U = \mathcal{A}_{(k_1,\dots,k_{\ell})}$.
(In other words, $U$ normalizes $\mathcal{A}_{(k_1,\dots,k_{\ell})}$.)
\end{itemize}

Furthermore, in a scenario where there are multiple unitaries acting on distinct registers, we impose an additional requirement that they commute with each other.
Thus, if $U$ acts on register $(\mathsf{R}_{k_1},\dots,\mathsf{R}_{k_{\ell}})$ and $V$ acts on register $(\mathsf{R}_{j_1},\dots,\mathsf{R}_{j_{m}})$, where 
$\{k_1,\dots,k_{\ell}\} \cap \{j_1,\dots,j_{m}\} = \emptyset$ then $U$ and $V$ commute.

\subsection{Definition of self-embezzlement in the commuting operator model}\label{sec:c-o-def}

There are four basic registers $\mathsf{A}_1,\mathsf{B}_1,\mathsf{A}_2,\mathsf{B}_2$ with respective C*-algebras $\mathcal{A}_1,\mathcal{B}_1,\mathcal{A}_2,\mathcal{B}_2 \subset B(\mathcal{H})$, where $\mathcal{H}$ is the underlying Hilbert space.
There is an initial state $\psi \in \mathcal{H}$ and unitary operations $U_A$ and $U_B$, in the reversible dynamics of registers of registers $(\mathsf{A}_1,\mathsf{A}_2)$ and $(\mathsf{B}_1,\mathsf{B}_2)$, respectively.
The final state is $\psi' = U_A U_B \psi$.

The following properties are relevant to being a self-embezzlement scheme:
\begin{enumerate}
\item
The initial state $\psi$ is a product state over the three registers $(\mathsf{A}_1,\mathsf{B}_1)$, 
$\mathsf{A}_2$ and $\mathsf{B}_2$.
Technically, this can be expressed as, for all $X \in \overline{\mathcal{A}_1 \cup \mathcal{B}_1}$, $Y \in \mathcal{A}_2$, and $Z \in \mathcal{B}_2$,
\begin{align}
\langle\psi, XYZ \psi\rangle =  \langle\psi, X \psi \rangle
\langle\psi, Y \psi\rangle
\langle\psi, Z \psi\rangle.
\end{align}

\item
Register $(\mathsf{A}_1,\mathsf{B}_1)$  incurs no net change when $U_AU_B$ is applied.
Technically, for all $X \in \overline{\mathcal{A}_1 \cup \mathcal{B}_1}$,
\begin{align}
\langle\psi, X \psi\rangle =  \langle\psi', X \psi'\rangle.
\end{align}

\item
For the final state $\psi'$, the bipartite state of register $(\mathsf{A}_1,\mathsf{B}_1)$ is the same as the bipartite state of register  $(\mathsf{A}_2,\mathsf{B}_2)$.
Technically, there exist unitary operations $W_A, W_B \in B(\mathcal{H})$ acting on registers $(\mathsf{A}_1,\mathsf{A}_2), (\mathsf{B}_1,\mathsf{B}_2)$ (respectively) that map between the two registers.
That is, $W_A^*\mathcal{A}_1 W_A = \mathcal{A}_2$ and $W_B^*\mathcal{B}_1 W_B = \mathcal{B}_2$.
And, for all $X \in \mathcal{A}_1$ and $Y \in \mathcal{B}_1$
\begin{align}
\langle\psi', X Y \psi'\rangle =  \langle\psi', (W_A^* X W_A)(W_B^* Y W_B)\psi' \rangle.
\end{align}
Thus, for $\psi'$, measurements of $(\mathsf{A}_1,\mathsf{B}_1)$ are equivalent to measurements of $(\mathsf{A}_2,\mathsf{B}_2)$, under the unitary transformation $W_AW_B$.
\end{enumerate}

\subsection{Converting from C*-circuit model to commuting-operator model}

We begin with an exact embezzlement protocol in the C*-circuit model $(\mathcal{A},\mathcal{B},\psi,\phi_A,\phi_B,\alpha_A,\alpha_B)$ from section~\ref{sec:yesgo}.
The C*-algebra of the entire system $(\mathsf{A}_1,\mathsf{B}_1,\mathsf{A}_2,\mathsf{B}_2)$ is
$\mathcal C = \cl A \otimes_{\min} \cl B \otimes_{\min} \cl A \otimes_{\min} \cl B$.
Define the initial state $s : \mathcal{C} \rightarrow \mathbb{C}$ as 
\begin{align}
s(x\otimes y \otimes z) = \psi(x)\phi_A(y)\phi_B(z),
\end{align}
for $x \in \mathcal{A}_1\otimes\mathcal{B}_1$, $y \in \mathcal{A}_2$, $z \in \mathcal{B}_2$.

The $*$-isomorphisms $\alpha_A$ and $\alpha_B$ are outer automorphisms of $\cl C$ of infinite order that commute with each other.
Using the $*$-crossed construction~\cite{Da}, we can extend $\cl C$ to $\mathcal{C} \rtimes ( \mathbb Z \times \mathbb Z)$, where the $*$-crossed product is with respect to the group action generated by $\alpha_A$ and $\alpha_B$.
We can regard the crossed-product 
$\cl C \rtimes ( \mathbb Z \times \mathbb Z)$ as the C*-algebra generated by $\cl C$ together with two unitaries, $u_A, u_B$ corresponding to the group elements $(1,0)$ and $(0,1)$, so that
\begin{align}
u_A^*(a_1 \otimes b_1 \otimes a_2 \otimes b_2) u_A 
&= a'_1 \otimes b_1 \otimes a'_2 \otimes b_2, 
\ \ \mbox{where \ $a'_1 \otimes a'_2 = \alpha_A(a_1 \otimes a_2)$},
\end{align}
and
\begin{align}
u_B^*(a_1 \otimes b_1 \otimes a_2 \otimes b_2)u_B 
&= a_1 \otimes b'_1 \otimes a_2 \otimes b'_2,
\ \ \mbox{where \ $b'_1 \otimes b'_2 = \alpha_B(b_1 \otimes b_2)$}.
\end{align}
From these equations one sees that $u_A$ commutes with all elements of the form $I \otimes b_1 \otimes I \otimes b_2$ and, similarly, $u_B$ commutes with all elements of the form $a_1 \otimes I \otimes a_2 \otimes I$.

The state $s$ extends to
$\cl C \rtimes (\mathbb Z \times \mathbb Z)$ by setting $s$ to 0 on all terms that contain a non-zero power of either $u_A$ or $u_B$. 
To see that this is a well-defined state, one uses the usual representation of the reduced crossed-product.

By applying the GNS Representation Theorem~\cite{GelfandN1943,Segal1947} to the state $s$, we obtain a Hilbert space $\mathcal{H}$, a unit vector $\eta \in \mathcal H$, and a unital $*$-homomorphism $\pi : \mathcal{C} \rtimes ( \mathbb Z \times \mathbb Z) \rightarrow B(\mathcal H)$ such that 
$s(c) = \langle \eta, \pi(c) \eta \rangle$ for all $c \in \mathcal{C}$.

Now we define 
\begin{align}
\cl A_1 &= \pi(\cl A \otimes I \otimes I \otimes I) \\
\cl A_2 &= \pi(I \otimes I \otimes \cl A \otimes I) \\
\cl B_1 &= \pi(I \otimes \cl B \otimes I \otimes I) \\
\cl B_2 &= \pi(I \otimes I \otimes I \otimes \cl B),
\end{align}
the initial state $\psi = \eta$, and 
\begin{align}
U_A &= \pi(u_A) \\
U_B &= \pi(u_B).
\end{align}
It is straightforward to check that conditions 1, 2, 3 in section~\ref{sec:c-o-def} hold for $\cl H$, $\mathsf{A}_1,\mathsf{B}_1,\mathsf{A}_2,\mathsf{B}_2$, $\psi$, $U_A$, and $U_B$.
To establish condition 3, we can set $W_A \bigl(u_A^{k} u_B^{\ell} (a_1 \otimes b_1 \otimes a_2 \otimes b_2)\big) = u_A^{k} u_B^{\ell} (a_2 \otimes b_1 \otimes a_1 \otimes b_2)$ and $W_B \bigl(u_A^{k} u_B^{\ell} (a_1 \otimes b_1 \otimes a_2 \otimes b_2)\big) = u_A^{k} u_B^{\ell} (a_1 \otimes b_2 \otimes a_2 \otimes b_1)$.

\section{Acknowledgments}

A substantial part of this paper was completed at l'Institut Henri Poincar\'{e} during their trimester ``Analysis in Quantum Information Theory" in the fall of 2017.
All authors met and collaborated there, and acknowledge its fruitful working environment.
We thank Debbie Leung, Miguel Navascu\'{e}s 
(who informed us of Ref.~\cite{NavascuesP2012}), William Slofstra, and Thomas Vidick (who informed us of Ref.~\cite{Kaniewski2016}) for helpful comments.
RC, VP, and LL acknowledge support by Canada's NSERC.
BC acknowledges support by Kakenhi 15KK0162, 17H04823, 17K18734 and ANR-14CE25-0003-01.

\appendix

\section{Some basics of C*-algebras}\label{sec:c-star-basics}

In this appendix, we supply a few basics of the theory of C*-algebras.

Given a Hilbert space $\cl H$, by a {\it concrete C*-subalgebra} $\cl A$, we mean any subset of the bounded linear operators on $\cl H, B(\cl H)$ that is a subalgebra, is a closed subset in the operator norm and has the property that if $A \in \cl A$, then its adjoint $A^*$(sometimes denoted $A^{\dag}$) is also in $\cl A$.  Such algebras also have an abstract characterisation.

\begin{defn} Let $\cl A$ be an algebra over the complex numbers, with a norm $\| \cdot \|$ and a map, $a \to a^*$, satisfying $(a+b)^* = a^* + b^*, \, (\lambda a)^* = \overline{\lambda} a^*, \, \forall \lambda \in \mathbb C$ and $(ab)^*= b^*a^*$. Then $\cl A$ is called a {\bf C*-algebra} provided:
\begin{itemize}
\item $\cl A$ is complete in the norm $\|\cdot\|$,
\item $\|ab\| \le \|a\| \|b\|$,
\item $\|a^*a\|= \|a\|^2.$
\end{itemize}
\end{defn}

A C*-algebra is {\bf unital} if it contains an identity element with respect to multiplication.
Given two C*-algebras, $\cl A$ and $\cl B$ a linear map $\pi: \cl A \to \cl B$ is called a {\bf $*$-homomorphism} provided that it is a homomorphism, i.e., $\pi(ab) = \pi(a) \pi(b)$, and $\pi(a^*) = \pi(a)^*.$ A $*$-homomorphism that is one-to-one and onto is called a {\bf $*$-isomorphism}, and a $*$-isomorphism from $\cl A$ to $\cl A$ is called a {\bf $*$-automorphism}.
A $*$-homomorphism is automatically contractive, and hence, a $*$-isomorphism is isometric.

By a {\bf state} on a unital C*-algebra $\cl A$, we mean any complex linear functional $s: \cl A \to \mathbb C$ that satisfies $s(I) =1$ and $s(a^*a) \ge 0$ for all $a \in \cl A$.
(Intuitively, any such $s$ defines the outcome probabilities of all possible POVM measurements: for a measurement element of the form $b$, where $0 \le b \le I$, $s(b)$ is the probability of outcome $b$.)

The celebrated {\bf Gelfand-Naimark-Segal theorem}~\cite{GelfandN1943,Segal1947} tells us that every abstract C*-algebra $\cl A$ is $*$-isomorphic to a concrete C*-subalgebra of $B(\cl H)$ for some Hilbert space. The key element of this theorem is a result about states.  
Given a unital C*-algebra $\cl A$ and a state, $s: \cl A \to \mathbb C$, its {\bf GNS representation} consists of a Hilbert space $\cl H_s$ and a $*$-homomorphism,
 \[ \pi_s: \cl A \to B(\cl H_s) \text{ and } \eta_s \in \cl H_s \text{ such that } s(a) = \langle \eta, \pi_s(a)  \eta_s\rangle.\] 
 Conversely, given a $\pi: \cl A \to B(\cl H)$ and unit vector $\eta$ we obtain a state by setting
  $s(a) = \langle \eta, \pi(a) \eta \rangle$.  We call this the {\it induced state}. The following identifies when a triple $(\pi, \cl H, \eta)$ is really the same as the GNS representation.
  
  Given a  representation $\pi:A \to B(\cl H)$ a unit vector $\eta \in \cl H$ is called \emph{cyclic} if the set of vectors  $\{ \pi(a) \eta: a \in \cl A \}$ is dense in $\cl H$.
  
  \begin{prop} Let $\pi: \cl A \to B(\cl H)$, let $\eta \in \cl H$ be a unit vector, let $s(a) = \langle \pi(a) \eta, \eta \rangle$ be the induced state and let $\pi_s: \cl A \to B(\cl H_s)$ and $\eta_s$ be the GNS representation.  If $\eta$ is cyclic, then there is a unitary $U: \cl H_s \to \cl H$ with $U \eta_s = \eta$ such that $U^*\pi(a) U = \pi_s(a)$ for all $a$.
  \end{prop}
    
    Let $S(\cl A)$ denote the set of all states on $\cl A$. This is a convex set and a state is \emph{pure} if and only if it is an extreme point of this set.
    
    \begin{prop} Let $s \in S(\cl A)$. The following are equivalent:
    \begin{enumerate}
    \item $s$ is pure,
    \item if $f: \cl A \to \mathbb C$ is a positive linear functional such that $f(p) \le s(p)$ for all $p \in \cl A^+$, then there is a constant $c \ge 0$ such that $f(a) = c\, s(a)$ for all $a$,
    \item if $\pi_s: \cl A \to B(\cl H_s)$ is the GNS representation, then $\pi_s(\cl H)^{\prime} = \mathbb C \cdot I_{\cl H_s}$.
  \end{enumerate}
  \end{prop}
  
  Because these three statements are the same, some books use one of these other properties as the definition of pure.  Combining the two results we see that:
  
  \begin{prop} Let $\pi: \cl A \to B(\cl H)$ be a representation and let $\eta \in \cl H$ be a cyclic vector. Then the induced state is pure if and only if $\pi(\cl A)^{\prime} = \mathbb C \cdot I_{\cl H}.$
  \end{prop}

\subsection{Tensor Products of C*-algebras}
In this paper we needed some properties of tensor products of C*-algebras, especially in defining the CAR algebra and its properties.  Here we provide a very brief explanation of some of these facts/ideas.

Let $\cl A$ and $\cl B$ be two unital C*-algebras, and let $\cl A\otimes \cl B$ be their algebraic tensor product. For $x= \sum_i a_i \otimes b_i$ and $y = \sum_j c_j \otimes d_j$ in $\cl A \otimes \cl B$ we set
\[ xy = \sum_{i,j} a_i c_j \otimes b_i d_j,\]
which defines a product, and we define {\it a *-map} by
\[x^* = \sum_i a_i^* \otimes b_i^*.\]
These two operations make $\cl A \otimes \cl B$ into a {\it $*$-algebra}.

Note that the *-subalgebra $\{ a \otimes 1: a \in \cl A \}$ can be identified with $\cl A$ and similarly, $\{ 1 \otimes b: b \in \cl B \}$ can be identified with $\cl B$. Also $(a \otimes 1)(1 \otimes b) = a \otimes b = (1 \otimes b)(1 \otimes a)$ so that these ``copies" of $\cl A$ and $\cl B$ commute.

There are two important ways to give this $*$-algebra a norm. Once we have a norm, it can be completed to become a C*-algebra. 

Given $x \in \cl A \otimes \cl B$ we set
\[\|x\|_{\max} = \sup \bigl\{ \| \pi(x) \|: \ \pi: \cl A \otimes \cl B \to B(\cl H) \text{ is a unital $*$-homomorphism}\bigr\}, \]
where the supremum is taken over all Hilbert spaces $\cl H$ and all unital $*$-homomorphisms, i.e., homomorphisms such that $\pi(x)^* = \pi(x^*)$.
The completion of $\cl A \otimes \cl B$ in this norm is a C*-algebra denoted $\cl A \otimes_{\max} \cl B$.

Alternatively, if $\pi_1: \cl A \to B(\cl H_1)$ and $\pi_2: \cl B \to B(\cl H_2)$ are unital $*$-homomorphisms, then setting $\pi(a \otimes b) = \pi_1(a) \otimes \pi_2(b) \in B(\cl H_1 \otimes \cl H_2)$ and extending linearly, defines a unital $*$-homomorphism from $\cl A \otimes \cl B$ into $B(\cl H_1 \otimes \cl H_2)$ denoted by $\pi= \pi_1 \otimes \pi_2$.

Given $x \in \cl A \otimes \cl B$ we set
\begin{multline*} \|x\|_{\min} = \sup \bigl\{ \| \pi_1 \otimes \pi_2(x) \|: \,\, \pi_1: \cl A \to B(\cl H_1), \, \, \pi_2: \cl B \to B(\cl H_2) \\ \text{ are unital $*$-homomorphisms}\bigr\}.\end{multline*}
The completion of $\cl A \otimes \cl B$ in this norm is a C*-algebra denoted $\cl A \otimes_{\min} \cl B$.

Some C*-algebras $\cl A$ have the property that for every C*-algebra $\cl B$ the min norm and the max norm on $\cl A \otimes \cl B$ are equal. Such algebras are called {\it nuclear}. For example, every matrix algebra is nuclear.

Once we have seen these definitions for pairs of algebras it is easy to see how to extend it to define a min and a max norm on the tensor product of any finite collection of C*-algebras. To extend these definitions to tensor products of infinitely many C*-algebras, one first defines the infinite algebraic tensor product of unital algebras to consist of elements that are formal infinite tensor products that are equal to the identity in all but finitely many components. This is the construction used to build the CAR algebra, which is known to be a nuclear C*-algebra.
 
 A good source for the above material is \cite{KR}.

  Here is an application of these ideas.
  
  \begin{thm} If $s_1: \cl A_1 \to \mathbb C$ and $s_2: \cl A_2 \to \mathbb C$ be pure states then the state $s_1 \otimes s_2: \cl A_1 \otimes_{\min} \cl A_2 \to \mathbb C$ defined as
\begin{align} 
s_1 \otimes s_2(a \otimes b) = s_1(a) s_2(b)
\end{align}
is also pure.
  \end{thm}
  \begin{proof} Let $\pi_i: \cl A_i \to B(\cl H_i)$, and $\eta_i$ be a GNS for $s_i, i=1,2$.  Then we have that  
  \[s_1 \otimes s_2(a \otimes b) = \langle \pi_1(a) \otimes \pi_2(b) \eta_1 \otimes \eta_2, \eta_1 \otimes \eta_2 \rangle.\]
  
  Given any vector in $u \in \cl H_1 \otimes \cl H_2$ it can be approximated by a finite sum $\sum_i h_i \otimes k_i$. But since $\eta_i$ are both cyclic vectors, we can approximate $h_i \sim \pi_1(a_i) \eta_1$ and $k_i \sim \pi_2(b_i) \eta_2$.  Thus,
  \[ u \sim \pi_1 \otimes \pi_2(\sum_j a_j \otimes b_j)(\eta_1 \otimes \eta_2).\]
  
  This proves that $\pi_1 \otimes \pi_2$ is the GNS representation for $s_1 \otimes s_2$.  So to show that it is pure we need to show that  $\big( \pi_1 \otimes \pi_2(\cl A_1 \otimes \cl A_2) \big)^{\prime} = \mathbb C \cdot (I \otimes I)$.
  
We can use operator matrices to do this.  If we fix a basis $\{ f_j \}$ for $\cl H_2$, then we can identify
  \[ \cl H_1 \otimes \cl H_2 \simeq \sum_j \cl H_1 \otimes f_j \simeq \oplus_j \cl H_1.\]
  With respect to this identification every $X \in B(\cl H_1 \otimes \cl H_2)$ is represented by a matrix,
  $X=(X_{i,j}),$ with $X_{i,j} \in B(\cl H_1)$.
  
  Note that  $\pi_1 \otimes \pi_2(a \otimes 1)$ becomes the diagonal matrix whose diagonal entries are $\pi_1(a)$.  Now for $X=(X_{i,j})$ in the commutant, we have
 \[(X_{i,j} \pi_1(a)) = X (\pi_1 \otimes \pi_2(a \otimes I)) = (\pi_1 \otimes \pi_2(a \otimes 1))X = (\pi_1(a) X_{i,j}).\]
 Thus, each $X_{i,j} \in \pi_1(\cl A_1)$ and since $s_1$ was pure, we have $X_{i,j} = \lambda_{i,j} I_{\cl H_1}$.
 
 So $X= I_{\cl H_1} \otimes T$ where $T= ( \lambda_{i,j})$ is its matrix representation with repsect to the onb $\{ f_j \}$.
 
 But now the fact that $X = I \otimes T$ commutes with every $I \otimes \pi_2(b)$ and the fact that $s_2$ is pure, implies that $T = \lambda I_{\cl H_2}$.
 \end{proof}
 
 The following results are good for showing that the states on the CAR algebra that we are interested in are pure.
 
 \begin{prop}\label{prop:pure}  Let $M_{n_1} \subseteq M_{n_2} \subseteq \cdots \subseteq \cl A$ be matrix algebras with $\cl A$ equal to the closure of their union, let $s: \cl A \to \mathbb C$ be a state and let $s_k: M_{n_k} \to \mathbb C$ be its restriction. If $s_k$ is pure for all $k$, then $s$ is pure.
 \end{prop}
 \begin{proof} We show that if $s$ is not pure, then there exists a $k$ so that $s_k$ is not pure.
 If $s$ is not pure, then there exist states $0 <t <1$ and states $\rho, \sigma$ on $\cl A$ such that $s = t \rho + (1-t) \sigma$ and $s \ne \rho$.
 
 If we let $\rho_k, \sigma_k$ be the restrictions to $M_{n_k}$ then $s_k = t \rho_k + (1-t) \sigma_k$. So we need to show that for some $k$,  $s_k \ne \rho_k$.
 But since $s \ne \rho$ there exists $a \in \cl A$ with $|s(a) - \rho(a)| = r>0$.  By density we can find  a $k$ and $a_k \in M_{n_k}$ with  $\|a- a_k\| < r/2 $.  Then
 \[ | s(a_k) - \rho(a_k)| = | (s - \rho)(a_k -a)  + (s- \rho)(a)| \ge |s(a) - \rho(a)| - |(s-\rho(a -a_k)| \ge r - 2\|a-a_k\| >0,\]
 hence $s_k \ne \rho_k$. 
  \end{proof}
  
  As an application consider our state on the CAR algebra defined by taking an infinite tensor product of $(\mathbb C^2, \psi_k)$.  At the $k$-th level this is the vector state induced by $\psi_1 \otimes \cdots \otimes \psi_k \in \mathbb C^{2^k}$ on  $M_{2^k}=B( \mathbb C^{2^k})$. Since $M_{2^k}$ contains all linear transformations on this vector space, the vector is cyclic so this is the GNS representation of the induced state.  But since the commutant of $M_{2^k} = B(\mathbb C^{2^k})$ is trivial, this state is pure.  Thus, since $s_k$ is pure for all $k$,  the induced state on the CAR algebra is pure. 

\subsection{Actions and crossed-products}
Some of our constructions use the concept of the crossed-product of a C*-algebra by an action of a group, which we briefly outline below. A good source for this material is \cite{Pedersen1979}.

Given a C*-algebra $\cl A$ and a unitary $u \in \cl A$, the map $a \to u^* a u$ is a $*$-automorphism of $\cl A$, but generally, not every $*$-automorphism can be obtained in this fashion.
However, given a collection of $*$-automorphisms of $\cl A$, there is always a C*-algebra $\cl B$, containing $\cl A$, such that each of these $*$-automorphism of $\cl A$ is given as conjugation by a unitary in $\cl B$. More precisely, let $Aut(\cl A)$ denote the group of $*$-automorphisms of $\cl A$, let $G$ be a (discrete) group, and let $\alpha: G \to Aut(\cl A),$ be a homomorphism.  Then there is a C*-algebra, denoted $\cl A \rtimes_{\alpha} G$ which is generated as an algebra by $\cl A$ and a set of unitaries, $U_g, \, g \in G$ with the properties that
\begin{itemize}
\item $\cl A \subseteq \cl A \rtimes_{\alpha} G$,
\item $U_g^* a U_g = \alpha_g(a), \forall a \in \cl A, g \in G,$
\item whenever $\cl B$ is a C*-algebra, $\pi: \cl A \to \cl B$ is a one-to-one $*$-homomorphism and there exist unitaries $V_g \in \cl B$ satisfying $V_g^* \pi(a) V_g = \pi(\alpha_g(a))$, then there is a $*$-homomorphism $\Pi: \cl A \rtimes_{\alpha} G \to \cl B$ with $\Pi(a) = \pi(a)$ for all $a \in \cl A$ and $\Pi(U_g) = V_g$, for all $g \in G$.
\end{itemize}

\end{document}